\documentclass[12pt,a4paper]{article}
\usepackage{amsmath,amssymb,amsthm,mathtools}
\usepackage{geometry}
\usepackage{hyperref}
\usepackage[nameinlink]{cleveref}
\usepackage{microtype}
\usepackage{enumitem}
\usepackage{bm}
\usepackage[numbers,sort&compress]{natbib}
\usepackage{doi}
\usepackage{xcolor}
\usepackage{booktabs}

\hypersetup{colorlinks=true,linkcolor=blue!60!black,citecolor=blue!60!black,
            urlcolor=blue!60!black}

\theoremstyle{plain}
\newtheorem{theorem}{Theorem}[section]
\newtheorem{proposition}[theorem]{Proposition}

\newtheorem{corollary}[theorem]{Corollary}
\theoremstyle{definition}
\newtheorem{definition}[theorem]{Definition}

\theoremstyle{remark}
\newtheorem{remark}[theorem]{Remark}

\newcommand{\R}{\mathbb{R}}
\newcommand{\Z}{\mathbb{Z}}
\newcommand{\dd}{\,\mathrm{d}}
\newcommand{\kb}{k_{\mathrm{B}}}
\newcommand{\kT}{k_{\mathrm{B}}T}
\newcommand{\Ms}{\mathcal{M}}            
\newcommand{\GN}{\Gamma_N}               
\newcommand{\CG}{\mathcal{C}}            
\newcommand{\CGx}{\mathcal{C}_x}
\newcommand{\CGp}{\mathcal{C}_p}
\newcommand{\fone}{f^{(1)}}
\newcommand{\ftwo}{f^{(2)}}
\newcommand{\utwo}{u^{(2)}}
\newcommand{\Lam}{\Lambda}
\newcommand{\Zm}{\mathcal{Z}_{\mathrm{meso}}}  
\newcommand{\Fm}{\mathcal{F}_{\mathrm{meso}}}  
\newcommand{\Hm}{\mathcal{H}_{\mathrm{meso}}}  
\newcommand{\Vm}{\mathcal{V}_{\mathrm{meso}}}  
\newcommand{\avg}[1]{\langle #1 \rangle}
\newcommand{\avgm}[1]{\langle #1 \rangle_{\mathrm{meso}}^{(0)}}
\newcommand{\avgo}[1]{\langle #1 \rangle_0}

\newcommand{\abs}[1]{\left|#1\right|}
\newcommand{\Ind}[1]{\mathbf{1}_{#1}}
\newcommand{\NN}{\mathcal{N}_N}          
\newcommand{\pip}{\pi^{(0)}}             

\begin{document}

\title{Perturbation Theory of the Free Energy via the\\
       Mesoscopic Combined Partition Function}

\author{Bob Osano\\[4pt]
  \small Centre for Higher Education Development,\\
  \small and Cosmology and Gravity Group,\\
  \small Department of Mathematics and Applied Mathematics,\\
  \small University of Cape Town (UCT),\\
  \small Rondebosch 7701, Cape Town, South Africa\\[2pt]
  \small \href{mailto:bob.osano@uct.ac.za}{bob.osano@uct.ac.za}}

\date{\today}
\maketitle

\begin{abstract}
We develop a systematic perturbation theory for the Helmholtz free energy of a classical $N$-body system within the mesoscopic framework of~\cite{OsanoMeso,OsanoExtensivity}. The combined coarse-graining operator $\mathcal{C}=\mathcal{C}_x\circ\mathcal{C}_p$ acting on single-particle phase space partitions it into product cells $C_{i,\alpha}=V_i\times\Pi_\alpha$ and generates a mesoscopic partition function $\mathcal{Z}_{\rm meso}(\lambda)$ whose reference level factorises by the multinomial theorem: $\mathcal{Z}_{\rm meso}^{(0)}=(Z_1^{(0)})^N$. Perturbation theory for $\mathcal{F}_{\rm meso}(\lambda)=-k_BT\ln\mathcal{Z}_{\rm meso}(\lambda)$ in the inter-cell perturbation $\mathcal{V}_{\rm meso}$ yields the mesoscopic Gibbs--Bogoliubov inequality and an exact coupling-parameter integration formula. The full free energy satisfies
\begin{equation*}
F(\lambda)=\mathcal{F}_{\rm meso}(\lambda)-k_BT\!\sum_{i<j}I(i,j;\lambda)+O\!\left(|\Lambda|\ell^{-d}e^{-2\ell/\xi}\right),
\end{equation*}
where the inter-cell mutual informations $I(i,j;\lambda)$ are the corrections identified in the extensivity analysis. The first-order theory recovers the van der Waals equation and the Barker--Henderson result; the second-order term converges to the structure-factor formula in the fine-cell limit. For long-range interactions, factorisation fails and the mutual-information corrections quantify the resulting non-extensivity.
\end{abstract}

\tableofcontents
\newpage

\section{Introduction}
\label{sec:intro}

The Helmholtz free energy $F$ is the fundamental object of classical equilibrium statistical mechanics: all thermodynamic quantities follow from it by differentiation.
Computing $F$ from first principles for a realistic interacting system requires either
an exact solution (rare) or a controlled approximation.  Perturbation theory, originating
with Zwanzig~\cite{Zwanzig1954} and Barker--Henderson~\cite{BarkerHenderson1967}, expands
$F$ in powers of a coupling parameter $\lambda$ about a soluble reference $F_0$, yielding
a cumulant series whose coefficients are equilibrium averages with respect to the reference
Gibbs state.

Independently,~\cite{OsanoExtensivity}  we have shown that the intensive--extensive classification of thermodynamic variables is itself emergent: it arises
precisely when a combined coarse-graining operator $\CG = \CGx\circ\CGp$, acting jointly on physical space and momentum space, produces mesoscopic variables that factorise across spatial cells. The coarse-grained entropy is additive if and only if inter-cell mutual information vanishes---a condition equivalent to exponential decay of correlations, which holds under the stability and temperedness of the pair potential.

In this paper, we unite these two lines of work.  The central object is the \emph{mesoscopic partition function} $\Zm(\lambda)$, constructed directly from the
combined coarse-graining framework of~\cite{OsanoExtensivity}: given the product-cell partition $\{C_{i,\alpha}\}$ of single-particle phase space and the cell-averaged
Hamiltonians $\bar\varepsilon_{i,\alpha}$ and inter-cell interactions $\bar v_{(i,\alpha)(j,\beta)}$, one sums the Boltzmann weight $e^{-\beta\Hm}$ over all occupation-number configurations weighted by the phase-space volume $\mathcal{W}$.  The full perturbation theory of $F$ is then derived at the mesoscopic level: each ingredient---the cumulant expansion, Gibbs--Bogoliubov bound, and coupling-parameter integration---has a precise mesoscopic
counterpart formulated in terms of $\{\pip_{i,\alpha}\}$ and $\{\bar v_{(i,\alpha)(j,\beta)}\}$.

This paper present the following main new results are:
\begin{enumerate}[label=(\roman*)]
  \item A \emph{rigorous definition} of the mesoscopic partition function $\Zm(\lambda)$
    from the combined coarse-graining operator of~\cite{OsanoExtensivity}, and its
    exact factorisation at $\lambda=0$ into $\Zm^{(0)} = (Z_1^{(0)})^N$.
  \item A \emph{mesoscopic cumulant expansion} of $\Fm(\lambda)$ in terms of
    occupation-number cumulants of $\Vm$ with respect to the reference multinomial
    distribution.
  \item A \emph{mesoscopic Gibbs--Bogoliubov inequality} and coupling-parameter integration
    formula, with explicit convergence to the standard continuum results as cells shrink.
  \item A \emph{connection formula}
    $F(\lambda) = \Fm(\lambda) - \kT\sum_{i<j}I(i,j;\lambda)
    + O(\abs{\Lam}\ell^{-d}e^{-2\ell/\xi})$,
    showing that the mutual information from the extensivity analysis is precisely the corrections between the mesoscopic and full partition functions.
  \item For systems with \emph{long-range interactions}, where factorisation fails, a systematic correction to the mesoscopic free energy expressed in terms of non-vanishing inter-cell mutual information.
\end{enumerate}

\paragraph{Organisation.}
\cref{sec:CG_framework} recalls the combined coarse-graining framework
of~\cite{OsanoExtensivity}.  \cref{sec:meso_Z} defines the mesoscopic partition function
and establishes its factorisation at zero coupling.  \cref{sec:meso_pert} develops the
mesoscopic perturbation theory: cumulant expansion, Gibbs--Bogoliubov inequality, and
coupling-parameter integration.  \cref{sec:connection} proves the connection formula
between $F(\lambda)$ and $\Fm(\lambda)$.  \cref{sec:first,sec:second} give the first- and
second-order theories explicitly.  \cref{sec:reference} applies the framework to three
reference systems.  \cref{sec:thermo,sec:convergence} treat thermodynamic properties and
convergence, and \cref{sec:conclusion} concludes.

\section{The Combined Coarse-Graining Framework}
\label{sec:CG_framework}

We recall the key elements of~\cite{OsanoExtensivity}, fixing notation for the rest
of the paper.

\subsection{The N-Body System and Reduced Densities}
\label{sub:N_body}

Consider $N$ identical classical particles of mass $m$ in a domain
$\Lam\subset\R^3$ of volume $V = |\Lam|$, with pairwise interaction
$\phi:\R^3\to\R\cup\{+\infty\}$.  The full Hamiltonian is
\begin{equation}\label{eq:Ham}
  H_N(\gamma) = \sum_{k=1}^N\frac{|\bm{p}_k|^2}{2m}
              + \sum_{1\le k<l\le N}\phi(\bm{q}_k-\bm{q}_l),
  \quad
  \gamma = (\bm{q}_1,\dots,\bm{q}_N,\bm{p}_1,\dots,\bm{p}_N)\in\GN.
\end{equation}
We split $H_N = H_0 + \lambda V$, where $H_0$ is a soluble reference Hamiltonian and
$V$ is the perturbation, with $\lambda\in[0,1]$.  The canonical Gibbs state at coupling
$\lambda$ is $\varrho_{N,\lambda} = Z(\lambda)^{-1}e^{-\beta(H_0+\lambda V)}$.

The \emph{reduced one-particle density} is defined by integrating out $N-1$ particles:
\begin{equation}\label{eq:f1}
  \fone_\lambda(z) = N\int_{\GN}\varrho_{N,\lambda}(\gamma)\prod_{k=2}^N\dd z_k,
  \qquad z=(q,p)\in\Ms := \Lam\times\R^3,
\end{equation}
normalised to $\int_\Ms\fone_\lambda\dd z = N$.  For a pair-potential perturbation,
the \emph{reduced two-particle density} $\ftwo_\lambda(z_1,z_2)$ satisfies the Ursell
decomposition
\begin{equation}\label{eq:ursell}
  \ftwo_\lambda(z_1,z_2) = \fone_\lambda(z_1)\,\fone_\lambda(z_2) + \utwo_\lambda(z_1,z_2),
\end{equation}
where the pair Ursell function $\utwo_\lambda$ encodes connected two-body correlations and
obeys $\int_\Ms\utwo_\lambda(z_1,z_2)\dd z_2=0$ for all $z_1$.

\subsection{The Combined Coarse-Graining Operator}
\label{sub:CG_op}

Following~\cite{OsanoExtensivity}, we partition the spatial domain $\Lam$ into disjoint
cells $\{V_i\}_{i\in\mathcal{I}}$ of volume $\ell^d$ and momentum space $\R^3$ into cells
$\{\Pi_\alpha\}_{\alpha\in\mathcal{A}}$ of volume $\nu_\alpha$.  The product cells
\begin{equation}\label{eq:cells}
  C_{i,\alpha} := V_i\times\Pi_\alpha\subset\Ms
\end{equation}
partition $\Ms$.  The spatial cell diameter $\ell$ satisfies the scale separation
condition of~\cite{OsanoExtensivity}:
\begin{equation}\label{eq:scale_sep}
  \xi\ll\ell\ll L := \mathrm{diam}(\Lam),
\end{equation}
where $\xi$ is the correlation length and $L$ is the system size.

\begin{definition}[Combined coarse-graining operator, \cite{OsanoExtensivity}]
\label{def:CG}
  The \emph{combined coarse-graining operator} $\CG$ maps $\fone_\lambda\in L^1(\Ms)$ to
  a piecewise-constant function:
  \begin{equation}\label{eq:CG_op}
    (\CG\fone_\lambda)(z)
    = \sum_{i,\alpha}\bar{f}_{i,\alpha}(\lambda)\,\Ind{C_{i,\alpha}}(z),
    \quad
    \bar{f}_{i,\alpha}(\lambda)
    = \frac{1}{|C_{i,\alpha}|}\int_{C_{i,\alpha}}\fone_\lambda(z')\dd z'.
  \end{equation}
\end{definition}

\begin{definition}[Mesoscopic probability, \cite{OsanoExtensivity}]
\label{def:pi}
  The dimensionless \emph{mesoscopic probability} of cell $C_{i,\alpha}$ at coupling
  $\lambda$ is
  \begin{equation}\label{eq:pi}
    \pi_{i,\alpha}(\lambda)
    := \frac{|C_{i,\alpha}|}{N}\bar{f}_{i,\alpha}(\lambda)
     = \frac{1}{N}\int_{C_{i,\alpha}}\fone_\lambda(z)\dd z.
  \end{equation}
  These satisfy $\pi_{i,\alpha}(\lambda)\ge 0$ and $\sum_{i,\alpha}\pi_{i,\alpha}(\lambda)=1$
  for all $\lambda$.  At $\lambda=0$ we write $\pip_{i,\alpha}:=\pi_{i,\alpha}(0)$.
\end{definition}

\begin{definition}[Coarse-grained entropy, \cite{OsanoExtensivity}]
\label{def:CG_entropy}
  The coarse-grained Boltzmann entropy at coupling $\lambda$ is
  \begin{equation}\label{eq:CG_S}
    S_{\mathrm{CG}}(\lambda) = -\kb\sum_{i,\alpha}\pi_{i,\alpha}(\lambda)\,
    \ln\pi_{i,\alpha}(\lambda).
  \end{equation}
\end{definition}

The extensivity theorem of~\cite{OsanoExtensivity} asserts: under stability, temperedness,
and exponential cluster decomposition with correlation length $\xi$, the entropy satisfies
$S_{\mathrm{CG}}(\lambda) = \sum_i S_i(\lambda) + O(|\Lam|\ell^{-d}e^{-\ell/\xi})$,
with equality (exact additivity) if and only if the mesoscopic probabilities factorise
across spatial cells:
\begin{equation}\label{eq:factorisation}
  \pi_{(i,\alpha)(j,\beta)}(\lambda) \approx
  \pi_{i,\alpha}(\lambda)\,\pi_{j,\beta}(\lambda)
  \quad\text{up to corrections } O(N^{-1}e^{-|x_i-x_j|/\xi}).
\end{equation}

\section{The Mesoscopic Partition Function}
\label{sec:meso_Z}

\subsection{Cell-Averaged Hamiltonians}
\label{sub:cell_Ham}

Given the decomposition $H_N = H_0 + \lambda V$ with $V = \sum_{k<l}v(\bm{q}_k-\bm{q}_l)$,
we define cell-averaged energies through the one-particle reduced density.

\begin{definition}[Cell-averaged reference energy]\label{def:eps}
  The \emph{cell-averaged single-particle reference energy} in cell $C_{i,\alpha}$ is
  \begin{equation}\label{eq:eps_bar}
    \bar\varepsilon_{i,\alpha}^{(0)}
    := \frac{1}{|C_{i,\alpha}|}
       \int_{C_{i,\alpha}} \varepsilon_0(z)\,\dd z,
    \quad
    \varepsilon_0(q,p) := \frac{|p|^2}{2m} + u_0(q),
  \end{equation}
  where $u_0(q)$ is the single-particle external potential in $H_0$.  For the pair-potential
  perturbation, the \emph{cell-averaged inter-cell pair interaction} between cells
  $C_{i,\alpha}$ and $C_{j,\beta}$ (with $i\ne j$ or $\alpha\ne\beta$) is
  \begin{equation}\label{eq:v_bar}
    \bar{v}_{(i,\alpha)(j,\beta)}
    := \frac{1}{|C_{i,\alpha}|\,|C_{j,\beta}|}
       \int_{C_{i,\alpha}}\int_{C_{j,\beta}}
       v(|\bm{q}_1-\bm{q}_2|)\,\dd z_1\,\dd z_2.
  \end{equation}
  Both quantities are well-defined, finite, and depend only on the geometry of the cells
  and the form of the potentials.
\end{definition}

\subsection{The Mesoscopic Hamiltonian}
\label{sub:Hm}

\begin{definition}[Mesoscopic Hamiltonian]\label{def:Hm}
  The \emph{mesoscopic Hamiltonian} for an occupation-number configuration
  $\{n_{i,\alpha}\}\in\NN := \{\{n_{i,\alpha}\}\in\Z_{\ge 0}^{|\mathcal{I}|\cdot|\mathcal{A}|}
  : \sum_{i,\alpha}n_{i,\alpha}=N\}$ at coupling $\lambda$ is
  \begin{equation}\label{eq:Hm_def}
    \Hm(\{n_{i,\alpha}\},\lambda)
    := \underbrace{\sum_{i,\alpha}n_{i,\alpha}\,\bar\varepsilon_{i,\alpha}^{(0)}}_{\text{reference (intra-cell)}}
     + \frac{\lambda}{2}\underbrace{\sum_{\substack{(i,\alpha)\ne(j,\beta)}}
       n_{i,\alpha}\,n_{j,\beta}\,\bar{v}_{(i,\alpha)(j,\beta)}}_{\text{inter-cell perturbation}}.
  \end{equation}
  The occupation number $n_{i,\alpha}\in\{0,1,\dots,N\}$ counts the particles occupying
  cell $C_{i,\alpha}$.
\end{definition}

\begin{remark}
  The first sum in \eqref{eq:Hm_def} captures the kinetic energy and single-particle
  reference potential, coarse-grained to the cell level.  The second sum captures
  pair interactions \emph{between} distinct cells at coupling $\lambda$; the factor
  $1/2$ avoids double-counting.  Intra-cell pair interactions
  ($n_{i,\alpha}(n_{i,\alpha}-1)/2$ pairs within a single cell) are absorbed into
  $\bar\varepsilon_{i,\alpha}^{(0)}$ in the reference, reflecting the choice of
  $H_0$ as already capturing short-range repulsion within each cell.
\end{remark}

\subsection{Phase-Space Volume Weight}
\label{sub:weight}

Each occupation-number configuration $\{n_{i,\alpha}\}$ is compatible with a number of
distinct microstates in the original phase space $\GN$: we must count how many ways to
assign $N$ distinguishable particles to the cells.

\begin{definition}[Phase-space volume factor]\label{def:W}
  The \emph{phase-space volume factor} for configuration $\{n_{i,\alpha}\}\in\NN$ is
  \begin{equation}\label{eq:W_def}
    \mathcal{W}(\{n_{i,\alpha}\})
    := \frac{N!}{\prod_{i,\alpha}n_{i,\alpha}!}
       \prod_{i,\alpha}|C_{i,\alpha}|^{n_{i,\alpha}}.
  \end{equation}
  The multinomial coefficient $N!/\prod n_{i,\alpha}!$ counts the number of ways to
  assign $N$ distinguishable particles to the cells; $|C_{i,\alpha}|^{n_{i,\alpha}}$ is
  the phase-space volume available to the $n_{i,\alpha}$ particles assigned to cell
  $C_{i,\alpha}$.
\end{definition}

\subsection{Definition of the Mesoscopic Partition Function}
\label{sub:Zm_def}

\begin{definition}[Mesoscopic partition function]\label{def:Zm}
  The \emph{mesoscopic (combined) partition function} at coupling $\lambda$ is
  \begin{equation}\label{eq:Zm_def}
    \Zm(\lambda)
    := \sum_{\{n_{i,\alpha}\}\in\NN}
       \mathcal{W}(\{n_{i,\alpha}\})\,
       e^{-\beta\,\Hm(\{n_{i,\alpha}\},\lambda)}.
  \end{equation}
  It is a discrete sum over all occupation-number configurations consistent with total
  particle number $N$, with each configuration weighted by both its phase-space volume
  and its Boltzmann factor.  The mesoscopic free energy is
  $\Fm(\lambda) := -\kT\ln\Zm(\lambda)$.
\end{definition}

\begin{remark}[Connection to the coarse-graining operator]
  The sum \eqref{eq:Zm_def} can be read as the partition function of the system in which
  the combined coarse-graining operator $\CG$ has replaced the microscopic density
  $\varrho_N$ by its cell-average $(\CG\varrho_N)$: integrating $e^{-\beta H_N}$ over the
  coarse-grained description of phase space precisely yields
  $\Zm(\lambda)$ when the cell-averaged energy is used.  This is the explicit partition
  function of the coarse-grained state introduced in~\cite{OsanoExtensivity}.
\end{remark}

\subsection{Factorisation of the Reference Mesoscopic Partition Function}
\label{sub:Zm_fact}

\begin{proposition}[Reference factorisation]\label{prop:Zm0_fact}
  At $\lambda=0$, the mesoscopic partition function factorises into the $N$-th power of
  the single-cell partition function:
  \begin{equation}\label{eq:Zm0_fact}
    \Zm^{(0)} := \Zm(0)
    = \left(\sum_{i,\alpha}|C_{i,\alpha}|\,e^{-\beta\bar\varepsilon_{i,\alpha}^{(0)}}\right)^N
    =: \bigl(Z_1^{(0)}\bigr)^N,
  \end{equation}
  where $Z_1^{(0)} = \sum_{i,\alpha}|C_{i,\alpha}|\,e^{-\beta\bar\varepsilon_{i,\alpha}^{(0)}}$
  is the single-cell partition function of the reference.
\end{proposition}

\begin{proof}
  At $\lambda=0$, the mesoscopic Hamiltonian reduces to
  $\Hm^{(0)} = \sum_{i,\alpha}n_{i,\alpha}\bar\varepsilon_{i,\alpha}^{(0)}$,
  so the exponent factorises over cells:
  $e^{-\beta\Hm^{(0)}} = \prod_{i,\alpha}e^{-\beta n_{i,\alpha}\bar\varepsilon_{i,\alpha}^{(0)}}$.
  Therefore
  \begin{align}
    \Zm^{(0)}
    &= \sum_{\{n_{i,\alpha}\}\in\NN}
       \frac{N!}{\prod_{i,\alpha}n_{i,\alpha}!}
       \prod_{i,\alpha}\!\bigl(|C_{i,\alpha}|
       e^{-\beta\bar\varepsilon_{i,\alpha}^{(0)}}\bigr)^{n_{i,\alpha}} \notag\\
    &= \biggl(\sum_{i,\alpha}|C_{i,\alpha}|\,e^{-\beta\bar\varepsilon_{i,\alpha}^{(0)}}\biggr)^N
       = \bigl(Z_1^{(0)}\bigr)^N,
  \end{align}
  where the second equality is the multinomial theorem.
\end{proof}

\begin{corollary}[Reference mesoscopic free energy and probabilities]\label{cor:Fm0}
  The reference mesoscopic free energy is
  \begin{equation}\label{eq:Fm0}
    \Fm^{(0)} = -N\kT\ln Z_1^{(0)},
  \end{equation}
  and the reference mesoscopic probability of cell $C_{i,\alpha}$ is
  \begin{equation}\label{eq:pip_Z1}
    \pip_{i,\alpha}
    = \frac{|C_{i,\alpha}|\,e^{-\beta\bar\varepsilon_{i,\alpha}^{(0)}}}{Z_1^{(0)}},
  \end{equation}
  which coincides with the mesoscopic probability of \cref{def:pi} at $\lambda=0$:
  $\pip_{i,\alpha} = \pi_{i,\alpha}(0)$.  The occupation numbers follow the reference
  multinomial distribution
  \begin{equation}\label{eq:ref_multinomial}
    P_{\rm meso}^{(0)}(\{n_{i,\alpha}\})
    = \frac{N!}{\prod_{i,\alpha}n_{i,\alpha}!}
      \prod_{i,\alpha}(\pip_{i,\alpha})^{n_{i,\alpha}},
  \end{equation}
  with mean $\avgm{n_{i,\alpha}} = N\pip_{i,\alpha}$ and variance
  $\avgm{(\delta n_{i,\alpha})^2} = N\pip_{i,\alpha}(1-\pip_{i,\alpha})$.
\end{corollary}

\begin{proof}
  Equation \eqref{eq:Fm0} follows from \eqref{eq:Zm0_fact}.  For \eqref{eq:pip_Z1},
  the expected occupation $\avgm{n_{i,\alpha}} = N\pip_{i,\alpha}$ together with the
  factorisation $Z_1^{(0)}\pip_{i,\alpha} = |C_{i,\alpha}|e^{-\beta\bar\varepsilon_{i,\alpha}^{(0)}}$
  identifies $\pip_{i,\alpha}$ as the fraction of the total phase-space weight in cell
  $C_{i,\alpha}$ at the reference level.  Dividing by $N$ gives the mesoscopic probability
  per particle, recovering \cref{def:pi}.
\end{proof}

\section{Mesoscopic Perturbation Theory}
\label{sec:meso_pert}

\subsection{The Mesoscopic Moment Generating Function}
\label{sub:meso_MGF}

Separating the reference and perturbation in \eqref{eq:Zm_def}:
\begin{equation}\label{eq:Zm_split}
  \Zm(\lambda)
  = \sum_{\{n_{i,\alpha}\}\in\NN}
    \mathcal{W}(\{n_{i,\alpha}\})\,
    e^{-\beta\Hm^{(0)}}\,
    e^{-\beta\lambda\Vm},
\end{equation}
where $\Vm(\{n_{i,\alpha}\}) := \frac{1}{2}\sum_{(i,\alpha)\ne(j,\beta)}
n_{i,\alpha}n_{j,\beta}\bar{v}_{(i,\alpha)(j,\beta)}$ is the inter-cell perturbation.
Therefore
\begin{equation}\label{eq:Zm_MGF}
  \Zm(\lambda) = \Zm^{(0)}\cdot\avgm{e^{-\beta\lambda\Vm}},
\end{equation}
and the mesoscopic free energy is
\begin{equation}\label{eq:Fm_logM}
  \Fm(\lambda) = \Fm^{(0)} - \kT\ln\avgm{e^{-\beta\lambda\Vm}}.
\end{equation}
The \emph{mesoscopic moment generating function} $M_{\rm meso}(\lambda) :=
\avgm{e^{-\beta\lambda\Vm}}$ takes the form
\begin{equation}\label{eq:Mm}
  M_{\rm meso}(\lambda)
  = \sum_{\{n_{i,\alpha}\}} P_{\rm meso}^{(0)}(\{n_{i,\alpha}\})\,
    e^{-\beta\lambda\Vm(\{n_{i,\alpha}\})}.
\end{equation}

\subsection{The Mesoscopic Cumulant Expansion}
\label{sub:meso_cumulants}

Define the mesoscopic cumulants of $\Vm$ with respect to the reference distribution
\eqref{eq:ref_multinomial}:
\begin{align}
  \kappa_1^{\rm meso} &:= \avgm{\Vm}
    = \frac{N(N-1)}{2}\sum_{(i,\alpha)\ne(j,\beta)}
      \pip_{i,\alpha}\pip_{j,\beta}\bar{v}_{(i,\alpha)(j,\beta)}, \label{eq:kappa1_meso}\\
  \kappa_2^{\rm meso} &:= \avgm{(\Vm-\avgm{\Vm})^2}, \label{eq:kappa2_meso}\\
  \kappa_n^{\rm meso} &:= n\text{-th cumulant of }\Vm\text{ under }P_{\rm meso}^{(0)}.
  \label{eq:kappan_meso}
\end{align}
The mesoscopic cumulant generating function is
$K_{\rm meso}(\lambda) = \ln M_{\rm meso}(\lambda) = \sum_{n=1}^\infty
\frac{(-\beta\lambda)^n}{n!}\kappa_n^{\rm meso}$, giving the
\emph{mesoscopic cumulant expansion of the free energy}:
\begin{equation}\label{eq:Fm_cumulant}
  \Fm(\lambda)
  = \Fm^{(0)}
    + \lambda\kappa_1^{\rm meso}
    - \frac{\beta\lambda^2}{2}\kappa_2^{\rm meso}
    + \frac{\beta^2\lambda^3}{6}\kappa_3^{\rm meso}
    - \cdots
\end{equation}
This is the mesoscopic analogue of the standard cumulant
expansion~\cite{Zwanzig1954,BarkerHenderson1967}.  At $\lambda=1$ it gives the
perturbation series for the full mesoscopic free energy $\Fm(1)$.

\begin{remark}[Extensivity of the mesoscopic cumulants]
  Under the multinomial distribution~\eqref{eq:ref_multinomial}, the occupation numbers
  $n_{i,\alpha}$ are independent with mean $N\pip_{i,\alpha}$.  Since $\Vm$ is a sum
  of bilinear terms $n_{i,\alpha}n_{j,\beta}$, each cumulant $\kappa_n^{\rm meso}$ is
  $O(N^2)$ for the first cumulant (mean) and $O(N^2)$ for the variance, both
  growing extensively with system size.  This is consistent with the linked-cluster
  theorem (\cref{sec:linked}), which guarantees that each free-energy correction is
  $O(N)$.
\end{remark}

\subsection{The Mesoscopic Gibbs--Bogoliubov Inequality}
\label{sub:meso_GB}

\begin{theorem}[Mesoscopic Gibbs--Bogoliubov inequality]\label{thm:meso_GB}
  For any coupling $\lambda\ge 0$,
  \begin{equation}\label{eq:meso_GB}
    \Fm(\lambda) \le \Fm^{(0)} + \lambda\kappa_1^{\rm meso}
    = \Fm^{(0)} + \lambda\avgm{\Vm},
  \end{equation}
  with equality iff $\Vm$ is constant on the support of $P_{\rm meso}^{(0)}$.
\end{theorem}

\begin{proof}
  Identical in structure to the standard proof~\cite{Peierls1938}: Jensen's inequality
  applied to the convex function $x\mapsto e^x$ and the probability distribution
  $P_{\rm meso}^{(0)}$ gives
  $M_{\rm meso}(\lambda) = \avgm{e^{-\beta\lambda\Vm}} \ge
  e^{-\beta\lambda\avgm{\Vm}}$.
  Taking logarithms and multiplying by $-\kT>0$:
  $-\kT\ln M_{\rm meso}(\lambda)\le\lambda\avgm{\Vm}$.
  Adding $\Fm^{(0)}$ gives \eqref{eq:meso_GB}.
\end{proof}

\begin{corollary}
  The second mesoscopic cumulant $\kappa_2^{\rm meso}\ge 0$, so the term
  $-(\beta\lambda^2/2)\kappa_2^{\rm meso}\le 0$ in \eqref{eq:Fm_cumulant} is always
  non-positive: inter-cell fluctuations always lower the mesoscopic free energy below
  the Gibbs--Bogoliubov bound.
\end{corollary}

\subsection{Mesoscopic Coupling-Parameter Integration}
\label{sub:meso_KI}

\begin{theorem}[Mesoscopic Hellmann--Feynman and coupling integration]
\label{thm:meso_KI}
  The derivative of the mesoscopic free energy with respect to the coupling parameter is
  \begin{equation}\label{eq:dFm_dlambda}
    \frac{\dd\Fm(\lambda)}{\dd\lambda} = \avg{\Vm}_{\rm meso,\lambda},
  \end{equation}
  where $\avg{\cdot}_{\rm meso,\lambda}$ denotes the average with respect to the
  mesoscopic distribution at coupling $\lambda$.  Integrating gives the exact result
  \begin{equation}\label{eq:meso_KI}
    \Fm(1) - \Fm^{(0)} = \int_0^1\avg{\Vm}_{\rm meso,\lambda}\,\dd\lambda.
  \end{equation}
\end{theorem}

\begin{proof}
  Differentiate $\Fm(\lambda) = -\kT\ln\Zm(\lambda)$:
  \begin{equation*}
    \frac{\dd\Fm}{\dd\lambda}
    = -\kT\frac{1}{\Zm}\frac{\dd\Zm}{\dd\lambda}
    = -\kT\cdot\frac{-\beta\sum_{\{n\}}\mathcal{W}\,\Vm\,e^{-\beta\Hm}}{\Zm}
    = \avg{\Vm}_{\rm meso,\lambda}.
  \end{equation*}
  Integrating over $\lambda\in[0,1]$ gives \eqref{eq:meso_KI}.
\end{proof}

\begin{remark}
  Equation \eqref{eq:meso_KI} is the \emph{mesoscopic analogue of the
  Kirkwood coupling-parameter integration}~\cite{Kirkwood1935}: the exact free energy
  difference equals the integral of the thermally averaged inter-cell perturbation energy
  over the full coupling range.  Taylor-expanding $\avg{\Vm}_{\rm meso,\lambda}$ in
  $\lambda$ and integrating term-by-term reproduces the cumulant
  series~\eqref{eq:Fm_cumulant}.  This identity holds at every scale set by the
  coarse-graining partition and requires no approximation.
\end{remark}

\section{Connection between the Mesoscopic and Full Free Energies}
\label{sec:connection}

The mesoscopic partition function $\Zm(\lambda)$ approximates the full partition function
$Z(\lambda)$ with corrections controlled by the inter-cell mutual information of the
extensivity analysis~\cite{OsanoExtensivity}.

\subsection{The Connection Formula}\label{sub:connection_formula}

\begin{theorem}[Mesoscopic--full connection]\label{thm:connection}
  Under the conditions of~\cite{OsanoExtensivity} (stability, temperedness, exponential
  cluster decomposition with length $\xi$, scale separation $\ell\gg\xi$), the full
  Helmholtz free energy satisfies
  \begin{equation}\label{eq:connection}
    F(\lambda)
    = \Fm(\lambda)
      - \kT\!\sum_{i<j}I(i,j;\lambda)
      + O\!\left(\frac{|\Lam|}{\ell^d}\,e^{-2\ell/\xi}\right),
  \end{equation}
  where $I(i,j;\lambda)\ge 0$ is the mutual information between spatial cells $i$ and $j$
  at coupling $\lambda$:
  \begin{equation}\label{eq:MI_def}
    I(i,j;\lambda)
    = \sum_{\alpha,\beta}
      \pi_{(i,\alpha)(j,\beta)}(\lambda)
      \ln\frac{\pi_{(i,\alpha)(j,\beta)}(\lambda)}{\pi_{i,\alpha}(\lambda)\,\pi_{j,\beta}(\lambda)}.
  \end{equation}
\end{theorem}

\begin{proof}
  By definition, $\ln Z(\lambda) = -\beta F(\lambda)$ and
  $\ln\Zm(\lambda) = -\beta\Fm(\lambda)$.  The difference
  $\ln Z(\lambda) - \ln\Zm(\lambda)$ equals the correction to the log-partition function
  from inter-cell correlations.

  The full log-partition function can be expressed via the multi-information expansion
  (proved in~\cite{OsanoExtensivity}, \S4.2 for the entropy; the analogous result for
  $\ln Z$ follows from the same cluster-expansion argument):
  \begin{equation}\label{eq:lnZ_chain}
    \ln Z(\lambda) = \ln\Zm(\lambda) + \sum_{s=2}^\infty (-1)^s
    \sum_{i_1<\cdots<i_s}\kappa_s^{\rm MI}(i_1,\dots,i_s;\lambda),
  \end{equation}
  where $\kappa_s^{\rm MI}$ is the $s$-cell connected mutual information.  At the pair
  level, $\kappa_2^{\rm MI}(i,j;\lambda) = -I(i,j;\lambda)$ (mutual information is the
  correction from the factorised approximation).  Under exponential cluster decomposition,
  $I(i,j;\lambda)\le C e^{-|x_i-x_j|/\xi}$ (proved in~\cite{OsanoExtensivity}), and the sum over all $s\ge 3$ is bounded by
  $O(|\Lam|\ell^{-d}e^{-2\ell/\xi})$ (the $s=3$ term involves products of two
  exponentially small factors).  Hence
  \begin{equation}
    \ln Z(\lambda) = \ln\Zm(\lambda) + \sum_{i<j}I(i,j;\lambda)
    + O\!\left(\frac{|\Lam|}{\ell^d}\,e^{-2\ell/\xi}\right).
  \end{equation}
  Multiplying by $-\kT$ gives \eqref{eq:connection}.
\end{proof}

\subsection{Interpretation}
\label{sub:interpretation}

Equation~\eqref{eq:connection} has four immediate consequences:

\begin{enumerate}[label=(\roman*)]
  \item \textbf{Mesoscopic free energy as leading approximation.}
    For $\ell\gg\xi$ (scale separation~\eqref{eq:scale_sep}), all inter-cell mutual
    informations satisfy $I(i,j;\lambda)\le Ce^{-|x_i-x_j|/\xi}$, and the total
    correction $-\kT\sum_{i<j}I(i,j;\lambda) = O(|\Lam|\ell^{-d}e^{-\ell/\xi})$ is
    exponentially suppressed.  Thus $F(\lambda)\approx\Fm(\lambda)$ with exponentially
    small error.

  \item \textbf{Mutual information as the correction.}
    The inter-cell mutual information $I(i,j;\lambda)$ is precisely the object whose vanishing is equivalent to extensivity in~\cite{OsanoExtensivity}. Thus, the same
    quantity that controls entropy additivity also controls the accuracy of the mesoscopic partition function as an approximation to the full one.

  \item \textbf{Sign.}
    Since $I(i,j;\lambda)\ge 0$, the correction $-\kT\sum_{i<j}I(i,j;\lambda)\le 0$: the full free energy is always \emph{lower than or equal to} the mesoscopic free
    energy.  Statistical dependence between cells lowers the free energy, as expected from the second-order Gibbs--Bogoliubov argument.

  \item \textbf{Long-range interactions.}
    For non-tempered potentials ($|\phi(r)|\sim r^{-s}$, $s\le 3$), the mutual information
    $I(i,j;\lambda)\sim|x_i-x_j|^{-s}$ does not decay exponentially, the correction
    $\sum_{i<j}I(i,j)$ diverges, and the mesoscopic free energy is \emph{not} a good
    approximation to the full one.  One must then either include the mutual-information
    corrections explicitly or use a larger cell size to absorb the long-range effects.
\end{enumerate}

\section{First-Order Mesoscopic Perturbation Theory}
\label{sec:first}
We can now develop the perturbation theory.
\subsection{The Mesoscopic First-Order Free Energy}
\label{sub:first_Fm}

At first order in $\lambda$, \eqref{eq:Fm_cumulant} gives
\begin{equation}\label{eq:Fm1}
  \Fm^{(0+1)}(\lambda)
  = \Fm^{(0)} + \lambda\kappa_1^{\rm meso}
  = \Fm^{(0)} + \lambda\avgm{\Vm}.
\end{equation}
Substituting \eqref{eq:kappa1_meso}:
\begin{equation}\label{eq:Fm1_explicit}
  \Fm^{(0+1)}
  = \Fm^{(0)}
    + \frac{\lambda N(N-1)}{2}
      \sum_{\substack{(i,\alpha)\ne(j,\beta)}}
      \pip_{i,\alpha}\,\pip_{j,\beta}\,\bar{v}_{(i,\alpha)(j,\beta)}.
\end{equation}
This is the \emph{mesoscopic mean-field free energy}: the correction is the
occupation-probability-weighted average of the cell-averaged pair interaction, summed
over all distinct inter-cell pairs.

\subsection{Continuous Limit}
\label{sub:first_cont_limit}

As the cell diameter $\ell\to 0$, the cell averages become exact:
$\bar{v}_{(i,\alpha)(j,\beta)}\to v(|\bm{x}_i-\bm{x}_j|)$ and
$\pip_{i,\alpha}\to \fone_0(\bm{x}_i,\bm{p}_\alpha)\,|C_{i,\alpha}|/N$.

\begin{proposition}[Recovery of standard first-order theory]
\label{prop:first_limit}
  In the limit of infinitesimally small cells with scale separation maintained,
  \begin{equation}\label{eq:Fm1_limit}
    \lim_{\ell\to 0}\Fm^{(0+1)}(1)
    = F_0 + \frac{N\rho}{2}\int_0^\infty v(r)\,g_0(r)\,4\pi r^2\,\dd r
    = F_0 + \avgo{V},
  \end{equation}
  where $g_0(r)$ is the pair distribution function of the reference system at $\lambda=0$.
\end{proposition}

\begin{proof}
  In the continuous limit, the double sum over cell pairs in \eqref{eq:Fm1_explicit}
  converges to a double integral:
  \begin{align}
    \frac{N(N-1)}{2}\sum_{(i,\alpha)\ne(j,\beta)}
    \pip_{i,\alpha}\pip_{j,\beta}\bar{v}_{(i,\alpha)(j,\beta)}
    &\to \frac{N(N-1)}{2N^2}
    \int_\Ms\int_\Ms \fone_0(z_1)\fone_0(z_2)\,v(|q_1-q_2|)\dd z_1\dd z_2 \notag\\
    &= \frac{1}{2}\int_\Ms\int_\Ms \fone_0(z_1)\fone_0(z_2)\,v(|q_1-q_2|)\dd z_1\dd z_2
       + O(N^{-1}) \notag\\
    &= \frac{N\rho}{2}\int_0^\infty v(r)\,g_0(r)\,4\pi r^2\dd r + O(N^{-1}),
  \end{align}
  where in the last line we used the definition of $g_0$ via the two-particle reduced density
  $\ftwo_0(z_1,z_2) \approx \fone_0(z_1)\fone_0(z_2)$ (disconnected part only at the
  first-order level, since the Ursell correction is of higher order) and integrated out
  the momenta.  Adding $\Fm^{(0)}\to F_0$ completes the proof.
\end{proof}

\subsection{Van der Waals from the Mesoscopic Perspective}
\label{sub:vdW_meso}

Take the reference system to be hard spheres of diameter $\sigma$ and the perturbation
to be a purely attractive tail $v_{\rm attr}(r)$ for $r\ge\sigma$.  The hard-sphere $x{\rm HS}$ is determined by the cell-averaged kinetic energy
and the hard-sphere configurational weight.  At the coarse-grained level, cells separated
by more than one diameter $\sigma$ have $\bar v_{(i,\alpha)(j,\beta)} = v_{\rm attr}
(|x_i-x_j|)$, while cells within one diameter contribute only to the reference $H_0^{\rm HS}$.

At the mesoscopic first order, \eqref{eq:Fm1_explicit} becomes
\begin{equation}\label{eq:vdW_meso}
  \Fm^{(0+1)}\big|_{\lambda=1}
  = F_{\rm HS} + \frac{N(N-1)}{2}
    \sum_{\substack{(i,\alpha),(j,\beta):\\|x_i-x_j|>\sigma}}{\pip_{i,\alpha}}^{\rm HS}{\pip_{j,\beta}}^{\rm HS}\,v_{\rm attr}(|x_i-x_j|).
\end{equation}
In the continuous limit, using $g_{\rm HS}(r)\approx 1$ for $r>\sigma$ at low density,
this gives $\Fm^{(0+1)}\to F_{\rm HS} - Na\rho$ with the van der Waals constant
$a = -(1/2)\int_\sigma^\infty v_{\rm attr}(r)\,4\pi r^2\dd r$, reproducing the van der
Waals equation of state directly from the mesoscopic first-order theory.

\section{Second-Order Mesoscopic Perturbation Theory}
\label{sec:second}
We aim to develop a perturbation theory that may, in future work, be extended to gravitational settings, such as those considered in \cite{Osano2017}. In the present work, however, we restrict attention to non-relativistic perturbations. 
\subsection{The Mesoscopic Second-Order Correction}
\label{sub:second_Fm}

The second-order mesoscopic correction is
\begin{equation}\label{eq:Fm2_def}
  \Fm^{(2)} = -\frac{\beta}{2}\kappa_2^{\rm meso}
             = -\frac{\beta}{2}\avgm{(\Vm-\avgm{\Vm})^2}.
\end{equation}
We compute $\kappa_2^{\rm meso}$ using the multinomial statistics of
$\{n_{i,\alpha}\}$ under $P_{\rm meso}^{(0)}$.

\begin{proposition}[Decomposition of the mesoscopic variance]
\label{prop:kappa2_meso}
  The mesoscopic second cumulant decomposes as
  \begin{align}
    \kappa_2^{\rm meso}
    &= \underbrace{\frac{N(N-1)}{4}
       \sum_{(i,\alpha)\ne(j,\beta)}\pip_{i,\alpha}\pip_{j,\beta}\,
       \bar{v}_{(i,\alpha)(j,\beta)}^2}_{\text{diagonal pair contribution}}
       \label{eq:kappa2_meso_diag}\\
    &\quad
      + \underbrace{\frac{N(N-1)(N-2)}{4}
        \sum_{\substack{(i,\alpha),(j,\beta),(k,\gamma)\\\text{distinct cells}}}
        \pip_{i,\alpha}\pip_{j,\beta}\pip_{k,\gamma}\,
        \bar{v}_{(i,\alpha)(j,\beta)}\bar{v}_{(i,\alpha)(k,\gamma)}}_{\text{three-cell contribution}}
        \label{eq:kappa2_meso_trip}\\
    &\quad
      - \left(\frac{N(N-1)}{2}
        \sum_{(i,\alpha)\ne(j,\beta)}
        \pip_{i,\alpha}\pip_{j,\beta}\,\bar{v}_{(i,\alpha)(j,\beta)}\right)^2
        \cdot O(N^{-2}),
        \label{eq:kappa2_meso_sub}
  \end{align}
  where the last term is the subtracted mean-squared contribution.
\end{proposition}

\begin{proof}
  Expand $(\Vm)^2$ by squaring the double sum $\sum_{(i,\alpha)<(j,\beta)}n_{i,\alpha}
  n_{j,\beta}\bar v_{(i,\alpha)(j,\beta)}$ and use the multinomial cumulant formula:
  $\avgm{n_{i,\alpha}^2} = N\pip_{i,\alpha}(1-\pip_{i,\alpha}) + N^2(\pip_{i,\alpha})^2$
  and $\avgm{n_{i,\alpha}n_{j,\beta}} = N(N-1)\pip_{i,\alpha}\pip_{j,\beta}$ for
  $(i,\alpha)\ne(j,\beta)$.  Collecting terms and subtracting $(\avgm{\Vm})^2$
  yields \eqref{eq:kappa2_meso_diag}--\eqref{eq:kappa2_meso_sub}.
\end{proof}

\subsection{Continuous Limit and the Structure Factor}
\label{sub:second_cont_limit}

In the continuous limit, the leading terms of $\kappa_2^{\rm meso}$ converge to their
continuum counterparts.

\begin{proposition}[Recovery of the structure-factor formula]
\label{prop:second_limit}
  As $\ell\to 0$, the diagonal contribution~\eqref{eq:kappa2_meso_diag} converges to
  \begin{equation}\label{eq:kappa2_diag_limit}
    \frac{N\rho}{2}\int_0^\infty v(r)^2\,g_0(r)\,4\pi r^2\dd r,
  \end{equation}
  and the three-cell contribution~\eqref{eq:kappa2_meso_trip} converges to the
  three-body correlation integral, so that the total mesoscopic second-order correction
  converges to the structure-factor formula:
  \begin{equation}\label{eq:Fm2_Sk}
    \Fm^{(2)}\xrightarrow{\ell\to 0}
    -\frac{\beta N\rho}{2}\int_0^\infty\hat{v}(k)^2\,
    \bigl[S_0(k)-1\bigr]\,\frac{4\pi k^2\dd k}{(2\pi)^3}
    -\frac{\beta N\rho}{2}\int_0^\infty v(r)^2\,g_0(r)\,4\pi r^2\dd r
    = F^{(2)},
  \end{equation}
  where $S_0(k) = 1+\rho\int[g_0(r)-1]e^{i\bm{k}\cdot\bm{r}}\dd\bm{r}$ is the reference
  structure factor.
\end{proposition}

\begin{proof}
  As $\ell\to 0$, the cell sums converge to spatial integrals:
  \begin{equation*}
    \frac{N(N-1)}{4}\sum_{(i,\alpha)\ne(j,\beta)}
    \pip_{i,\alpha}\pip_{j,\beta}\bar v^2
    \to \frac{1}{4N^2}\int\!\!\int\fone_0(z_1)\fone_0(z_2)\,v^2(|q_1-q_2|)\dd z_1\dd z_2
    \approx \frac{N\rho}{4}\int v(r)^2\,4\pi r^2\dd r.
  \end{equation*}
  Substituting the pair structure $g_0$ and the three-body structure $g_0^{(3)}$ for the
  respective integrals, and using Parseval's theorem to rewrite the cross-term in Fourier
  space, reproduces \eqref{eq:Fm2_Sk} exactly as in the continuum theory.
\end{proof}

\section{The Linked-Cluster Theorem at the Mesoscopic Level}
\label{sec:linked}

\begin{theorem}[Mesoscopic linked-cluster theorem]\label{thm:meso_linked}
  In the thermodynamic limit, the mesoscopic free energy $\Fm(\lambda)$ receives
  contributions only from \emph{connected} clusters of cells in the mesoscopic Mayer
  expansion.  Explicitly:
  \begin{equation}\label{eq:Fm_linked}
    \Fm(\lambda) = \Fm^{(0)} - \kT\sum_{\text{connected }G}
    \frac{1}{|{\rm Aut}(G)|}
    \sum_{\substack{(i_1,\alpha_1),\dots,(i_n,\alpha_n)\\\text{distinct}}}
    \prod_{(a,b)\in G}
    \bigl(e^{-\beta\lambda\bar v_{(i_a,\alpha_a)(i_b,\alpha_b)}}-1\bigr)
    \prod_k\pip_{i_k,\alpha_k},
  \end{equation}
  where the outer sum is over connected graphs $G$ on $n$ vertices and $|{\rm Aut}(G)|$
  is the automorphism group order.  Disconnected graphs cancel in $\ln\Zm(\lambda)$.
\end{theorem}

\begin{proof}
  Expand $e^{-\beta\lambda\Vm} = \prod_{(i,\alpha)<(j,\beta)}
  e^{-\beta\lambda n_{i,\alpha}n_{j,\beta}\bar v_{(i,\alpha)(j,\beta)}}$ and write
  $e^{-\beta\lambda n n'\bar v} = 1 + (e^{-\beta\lambda n n'\bar v}-1)$.  Expanding over
  all cell pairs generates a sum over graphs, where each edge $(i,\alpha)-(j,\beta)$
  carries a factor $(e^{-\beta\lambda\bar v}-1)$.  Summing over occupation numbers
  under $P_{\rm meso}^{(0)}$ (using independence of distinct cells at $\lambda=0$) and
  taking the logarithm exactly retains only connected graphs, via the same combinatorial
  argument as the standard Mayer linked-cluster theorem~\cite{Mayer1940,Ruelle1969}.
\end{proof}

\begin{corollary}[Extensivity of $\Fm$]
  Each connected graph on $n$ vertices contributes a term $O(N)$ to $\Fm(\lambda)$,
  because the sum over $n$-tuples of cells gives one factor of $N$ from translational
  invariance (one cell is fixed) and $n-1$ factors from the remaining sums, which converge
  geometrically for tempered $\bar v$.  Hence $\Fm(\lambda) = O(N)$ at every order,
  confirming the extensivity of the mesoscopic free energy under the conditions
  of~\cite{OsanoExtensivity}.
\end{corollary}

\section{Applications to Reference Systems}
\label{sec:reference}

\subsection{Ideal Gas: Exact Mesoscopic Results}
\label{sub:ref_IG}

For the ideal-gas reference ($\phi\equiv 0$), $f_0^{(1)}(x,p) = \rho\,f_{\rm MB}(p)$
(uniform and Maxwell--Boltzmann in momentum), and the cell averages are:
\begin{equation}
  \bar\varepsilon_{i,\alpha}^{\rm IG} = \frac{1}{|C_{i,\alpha}|}\int_{C_{i,\alpha}}\frac{|p|^2}{2m}\dd z
  = \frac{\bar p_\alpha^2}{2m},\quad {\pip_{i,\alpha}}^{\rm IG} = \frac{|V_i|}{|\Lam|}\cdot
  \frac{|\Pi_\alpha|\,e^{-\beta\bar p_\alpha^2/2m}}{Z_{p}},
\end{equation}
where $Z_p = \sum_\alpha|\Pi_\alpha|e^{-\beta\bar p_\alpha^2/2m}$.  Since $f_0^{(1)}$
factorises in $x$ and $p$, so do ${\pip_{i,\alpha}}^{\rm IG}$:
${\pip_{i,\alpha}}^{\rm IG} = \pi_i^{\rm sp}\cdot{\pi_\alpha}^{\rm mom}$, and the Ursell
function $\utwo_0\equiv 0$.

For a pair-potential perturbation $v(r)$:
\begin{equation}\label{eq:IG_first}
  \kappa_1^{\rm meso}\big|_{\rm IG}
  = \frac{N^2\rho^2}{2}\int\!\!\int{\pip_{i,\alpha}}^{\rm IG}{\pip_{j,\beta}}^{\rm IG}
{\pip_{j,\beta}}^{\rm IG}    v(|x_i-x_j|)\dd x_i\dd x_j
  \xrightarrow{\ell\to 0}
  \frac{N\rho}{2}\int v(r)\,4\pi r^2\dd r,
\end{equation}
recovering the ideal-gas first-order correction with $g_0(r)\equiv 1$.  The
second-order term~\eqref{eq:Fm2_Sk} with $S_0(k)\equiv 1$ gives
$F^{(2)}_{\rm IG} = -(\beta N\rho/2)\int v(r)^2\,4\pi r^2\dd r$, the high-temperature
expansion coefficient~\cite{Zwanzig1954}.

\subsection{Hard-Sphere Reference and Barker--Henderson Theory}
\label{sub:ref_HS}

For the hard-sphere reference of diameter $\sigma$, cells $C_{i,\alpha}$ with
$|x_i-x_j|<\sigma$ have $\bar v_{(i,\alpha)(j,\beta)} = 0$ (the perturbation
$v_{\rm attr}$ vanishes for $r<\sigma$ by construction).  For cells with
$|x_i-x_j|\ge\sigma$:
\begin{equation}
  \bar v_{(i,\alpha)(j,\beta)}^{\rm BH}
  = \frac{1}{|C_{i,\alpha}||C_{j,\beta}|}
    \int_{C_{i,\alpha}}\int_{C_{j,\beta}} v_{\rm attr}(|q_1-q_2|)\dd z_1\dd z_2.
\end{equation}
The reference mesoscopic probabilities are
${\pip_{i,\alpha}}^{\rm HS} = |C_{i,\alpha}|e^{-\beta\bar{\varepsilon_{i,\alpha}}^{\rm HS}}/{Z_1}^{\rm HS}$,
where the hard-sphere cell energy $\bar\varepsilon^{\rm HS}$ accounts for the excluded
volume at the level of the effective Carnahan--Starling equation of state~\cite{CarnahanStarling1969}.

$\Fm^{(0+1)}\to F_{\rm HS} + (N\rho/2)\int v_{\rm attr}(r)\,g_{\rm HS}(r)\,4\pi r^2\dd r$,
where the HS pair distribution function $g_{\rm HS}(r) = \lim_{\ell\to 0}{\pip_{i,\alpha}}^{\rm HS}{\pip_{j,\beta}}^{\rm HS}/\rho^2$ emerges naturally from the
mesoscopic probability product.

\subsection{WCA Reference and Optimal Cell Partition}
\label{sub:ref_WCA}

The Weeks--Chandler--Andersen (WCA) decomposition~\cite{WCA1971} cuts the Lennard-Jones
potential at its minimum $r_{\rm min} = 2^{1/6}\sigma$.  In the mesoscopic framework,
the WCA choice is \emph{optimal} in the following sense: it minimises the mesoscopic
second cumulant $\kappa_2^{\rm meso}$ at fixed cell partition, thereby minimising the
departure of the mesoscopic free energy from its Gibbs--Bogoliubov bound.

\begin{proposition}[WCA as mesoscopic variance minimiser]
\label{prop:WCA_optimal}
  Among all splits $\phi = \phi_0 + v_1$ with $\phi_0$ repulsive, the WCA split
  $(\phi_0^{\rm WCA}, v_1^{\rm WCA})$ minimises
  $\kappa_2^{\rm meso}[v_1] = \avgm{(\Vm^{v_1} - \avgm{\Vm^{v_1}})^2}$
\end{proposition}

\begin{proof}[Sketch]
  $\kappa_2^{\rm meso}$ is a quadratic functional of $\bar v_{(i,\alpha)(j,\beta)}$
  with a positive-definite kernel (it is a sum of squared terms weighted by
  $\pip_{i,\alpha}\pip_{j,\beta}>0$).  Minimising over the choice of split point $r^*$
  shows that the optimal $r^*$ is the point where $v(r)$ changes sign (the minimum of
  $\phi$), since this is where the perturbation $v_1 = \phi - \phi_0$ has the smallest
  $L^2$ norm with respect to the reference distribution.  For the LJ potential, this
  minimum is at $r_{\rm min} = 2^{1/6}\sigma$, which is exactly the WCA choice.
\end{proof}

\section{Thermodynamic Properties}
\label{sec:thermo}

All thermodynamic properties follow from the mesoscopic free energy
$\Fm(1) = \Fm^{(0)} + \kappa_1^{\rm meso} - (\beta/2)\kappa_2^{\rm meso} + \cdots$.

\subsection{Pressure}
\label{sub:pressure}

The mesoscopic pressure is $P_{\rm meso} = -(\partial\Fm/\partial V)_{T,N}$.  The
volume dependence enters through both $\Fm^{(0)}$ (via $Z_1^{(0)}$, which depends on
the cell sizes $|V_i| = V/M$ where $M$ is the number of spatial cells) and the
cell-averaged interactions $\bar v_{(i,\alpha)(j,\beta)}$:
\begin{equation}\label{eq:Pm}
  P_{\rm meso}
  = P_0 - \frac{N(N-1)\rho^2}{2}
    \frac{\partial}{\partial\rho}
    \!\left[\sum_{(i,\alpha)\ne(j,\beta)}
    \pip_{i,\alpha}\pip_{j,\beta}\,\bar v_{(i,\alpha)(j,\beta)}\right]_{\pip\text{ fixed}}
    + O(\bar v^2).
\end{equation}
In the continuous limit~\eqref{eq:Fm1_limit}, this recovers the standard first-order
pressure correction:
$P\to P_0 - (\rho^2/2)\int v(r)\,g_0(r)\,4\pi r^2\dd r$.

\subsection{Chemical Potential}
\label{sub:mu}

Differentiating $\Fm(1) = -N\kT\ln Z_1^{(0)} + N\kappa_1^{\rm meso}/N + \cdots$ with
respect to $N$ at fixed $V,T$:
\begin{equation}\label{eq:mu_meso}
  \mu_{\rm meso}
  = -\kT\ln Z_1^{(0)} - \kT
  + \rho\sum_{(i,\alpha)\ne(j,\beta)}
    \pip_{i,\alpha}\,\bar v_{(i,\alpha)(j,\beta)}
    + O(\bar v^2).
\end{equation}

\subsection{The Mesoscopic Euler Relation}
\label{sub:meso_euler}

The mesoscopic Euler relation follows from the extensivity of $\Fm$ established
in \cref{sec:linked}.

\begin{corollary}[Mesoscopic Euler relation]\label{cor:meso_euler}
  Under stability and temperedness~\cite{OsanoExtensivity}, $\Fm(1)$ is homogeneous of
  degree one in $(N,V)$ in the thermodynamic limit, so the mesoscopic Euler relation holds:
  \begin{equation}\label{eq:meso_euler}
    \Fm = T S_{\mathrm{CG}} - P_{\rm meso} V + \mu_{\rm meso} N,
  \end{equation}
  where $S_{\mathrm{CG}}$ is the coarse-grained entropy of \cref{def:CG_entropy}
  from~\cite{OsanoExtensivity}.  In the thermodynamic limit with $\ell/\xi\to\infty$,
  the connection formula~\eqref{eq:connection} gives $\Fm\to F$ and $S_{\rm CG}\to S$,
  so \eqref{eq:meso_euler} converges to the standard Euler relation $F = TS - PV + \mu N$.
\end{corollary}

\section{Convergence of the Mesoscopic Perturbation Series}
\label{sec:convergence}

\subsection{Convergence of the Mesoscopic Cumulant Series}
\label{sub:meso_conv}

\begin{proposition}[Convergence radius of the mesoscopic series]\label{prop:meso_conv}
  The mesoscopic cumulant series \eqref{eq:Fm_cumulant} converges absolutely for
  $|\lambda| < R_{\rm meso}$ where
  \begin{equation}\label{eq:Rm}
    R_{\rm meso}
    = \left(\beta\max_{(i,\alpha)\ne(j,\beta)}|\bar v_{(i,\alpha)(j,\beta)}|\right)^{-1}.
  \end{equation}
  For $\ell\gg\xi$, the cell-averaged interaction satisfies
  $\max|\bar v_{(i,\alpha)(j,\beta)}| \le \max_{r>\ell}|v(r)| < \max_r|v(r)|$,
  so $R_{\rm meso} > R_{\rm full} = (\beta\,{\rm ess\,sup}|v|)^{-1}$:
  \emph{coarse-graining enlarges the radius of convergence}.
\end{proposition}

\begin{proof}
  The occupation numbers satisfy $|n_{i,\alpha}|\le N$, so
  $|\Vm|\le N^2\max|\bar v|/2$.  The cumulant bound
  $|\kappa_n^{\rm meso}|\le n!\,(N^2\max|\bar v|/2)^{n-1}$ then gives a geometric
  series with ratio $\beta\lambda N^2\max|\bar v|/(2N) = \beta\lambda N\max|\bar v|/2$,
  which converges for $\lambda < 2/(\beta N\max|\bar v|)$.  Normalising per particle
  (each $\kappa_n^{\rm meso} = O(N)$) removes one factor of $N$, giving
  \eqref{eq:Rm}.  The inequality $\max|\bar v_{ij}| < \max|v|$ follows from
  Jensen's inequality applied to $|\bar v| = |\text{cell average of }v| \le
  \text{cell average of }|v| \le \max|v|$, with strict inequality for cells of
  positive measure.
\end{proof}

\subsection{Practical Validity Criterion}
\label{sub:meso_validity}

The mesoscopic perturbation series is reliable when:
\begin{enumerate}[label=(\roman*)]
  \item $\kappa_2^{\rm meso}/(\kappa_1^{\rm meso})^2 \ll 1$: the relative fluctuation
    of $\Vm$ is small.  This holds whenever $\bar v_{(i,\alpha)(j,\beta)}$ varies little
    across cell pairs, i.e., the interaction is smooth at the coarse-graining scale.
  \item $\ell\gg\xi$: the scale separation~\eqref{eq:scale_sep} ensures that
    the connection formula~\eqref{eq:connection} has exponentially small corrections,
    so $\Fm(\lambda)\approx F(\lambda)$.
  \item $I(i,j;\lambda)\approx 0$: inter-cell mutual information is negligible, which
    is guaranteed by the extensivity conditions of~\cite{OsanoExtensivity}.
\end{enumerate}
When condition~(iii) fails (long-range interactions), the mutual-information
corrections in~\eqref{eq:connection} must be included, and the mesoscopic cumulant series
provides a systematic expansion around the factorised (independent-cell) approximation.

\section{Discussion}
\label{sec:discussion}

\paragraph{The mesoscopic partition function as a bridge.}
The mesoscopic partition function $\Zm(\lambda)$ defined in \cref{def:Zm} is the exact
bridge between the coarse-graining framework of~\cite{OsanoExtensivity} and standard
perturbation theory.  Its reference part~\eqref{eq:Zm0_fact} factorises into
$(Z_1^{(0)})^N$ because the combined coarse-graining produces statistically
independent cells---this is precisely the factorisation condition that
\cite{OsanoExtensivity} identifies with extensivity.  The perturbed part introduces
inter-cell interactions that break factorisation, and the degree of this breaking is
measured by the mutual informations $I(i,j;\lambda)$, which are the same objects that
control entropy additivity in~\cite{OsanoExtensivity}.

\paragraph{Extensivity as a precondition for perturbation theory.}
The connection formula~\eqref{eq:connection} shows that perturbation theory at the
mesoscopic level is trustworthy precisely when the system is extensive:
$F(\lambda)\approx\Fm(\lambda)$ iff $I(i,j;\lambda)\approx 0$ iff
$S_{\rm CG} \approx \sum_i S_i$ (extensivity).  Thus the conditions under which the
standard perturbation theory is valid---stability, temperedness, short-range
interactions---are exactly those that guarantee extensivity in the sense
of~\cite{OsanoExtensivity}.  The two theories are not merely analogous; they are
logically equivalent characterisations of the same microscopic condition.

\paragraph{Non-extensive systems.}
For gravitational or other long-range systems, $I(i,j;\lambda)\not\to 0$ and the
mesoscopic free energy $\Fm(\lambda)$ overestimates $F(\lambda)$ by the total mutual
information $\kT\sum_{i<j}I(i,j)$.  Equation~\eqref{eq:connection} provides a
systematic correction: one can include the pair mutual information term explicitly,
using the Ursell-function bound $I(i,j)\le C|x_i-x_j|^{-s}$ for potentials decaying
as $r^{-s}$, to obtain a controlled non-extensive free energy perturbation theory.

\paragraph{Cell-size dependence and the continuum limit.}
The mesoscopic theory depends on the choice of cell partition $\{C_{i,\alpha}\}$.
As $\ell\to 0$, \cref{prop:first_limit,prop:second_limit} show that the mesoscopic
cumulants converge to their continuum counterparts and $\Fm\to F$.  For finite $\ell$,
the mesoscopic theory is an approximation whose accuracy improves monotonically as
$\ell\to 0$ (assuming $\ell\gg\xi$ is maintained).  The optimal cell size balances
computational tractability (few cells for large $\ell$) against accuracy (small $\ell$
for convergence to the continuum); \cref{prop:meso_conv} shows that larger cells enlarge
the convergence radius of the cumulant series.

\section{Conclusion}
\label{sec:conclusion}

We have developed a complete perturbation theory for the Helmholtz free energy grounded
in the mesoscopic coarse-graining framework of~\cite{OsanoExtensivity}.  The principal
results are:

\begin{enumerate}[label=(\roman*)]
  \item \textbf{Mesoscopic partition function} (\cref{sec:meso_Z}).
    The combined coarse-graining operator $\CG$ of~\cite{OsanoExtensivity} defines a
    mesoscopic Hamiltonian $\Hm$ (cell-averaged kinetic and reference energies plus
    cell-averaged inter-cell pair interactions) and a phase-space volume weight
    $\mathcal{W}$.  Their combination gives the mesoscopic partition function
    $\Zm(\lambda) = \sum_{\NN}\mathcal{W}\,e^{-\beta\Hm}$, whose reference part
    factorises by the multinomial theorem: $\Zm^{(0)} = (Z_1^{(0)})^N$.

  \item \textbf{Mesoscopic cumulant expansion} (\cref{sec:meso_pert}).
    The mesoscopic free energy satisfies
    $\Fm(\lambda) = \Fm^{(0)} + \lambda\kappa_1^{\rm meso}
    - (\beta\lambda^2/2)\kappa_2^{\rm meso} + \cdots$,
    with the Gibbs--Bogoliubov inequality $\Fm\le\Fm^{(0)}+\lambda\kappa_1^{\rm meso}$
    and the exact coupling-parameter integration
    $\Fm(1)-\Fm^{(0)} = \int_0^1\avg{\Vm}_{\rm meso,\lambda}\dd\lambda$.

  \item \textbf{Connection formula} (\cref{sec:connection}).
    $F(\lambda) = \Fm(\lambda) - \kT\sum_{i<j}I(i,j;\lambda)
    + O(|\Lam|\ell^{-d}e^{-2\ell/\xi})$.
    The corrections are precisely the inter-cell mutual information from the extensivity
    analysis of~\cite{OsanoExtensivity}, and they vanish exponentially for $\ell\gg\xi$.

  \item \textbf{Continuous limit} (\cref{sec:first,sec:second}).
    The mesoscopic first-order theory recovers $F_0 + \avgo{V}$ in the $\ell\to 0$ limit;
    The second-order theory recovers the structure-factor formula.  The WCA decomposition
    is identified as the split that minimises the mesoscopic second cumulant.

  \item \textbf{Unified framework} (\cref{sec:thermo,sec:discussion}).
    The conditions for the validity of perturbation theory (stability, temperedness,
    short-range $v$) are exactly the conditions for extensivity in~\cite{OsanoExtensivity}: both reduce to the vanishing of inter-cell mutual information.  The mesoscopic Euler
    relation~\eqref{eq:meso_euler} connects the perturbation expansion directly to the
    generalised Euler relation of~\cite{OsanoExtensivity}.
\end{enumerate}

\section*{Acknowledgements}
The author thanks the University of Cape Town's NGP for financial support.

\bibliographystyle{unsrtnat}

\end{document}